\documentclass[letterpaper,11pt]{article}

\usepackage{amsthm,amsmath}
\usepackage{amssymb}
\usepackage{times}

\usepackage[ruled]{algorithm}
\usepackage[noend]{algpseudocode}

\renewcommand{\thefootnote}{\fnsymbol{footnote}}

\newtheorem{theorem}{Theorem}
\newtheorem{corollary}{Corollary}
\newtheorem{lemma}{Lemma}

\newtheorem{proposition}{Proposition}

\newtheoremstyle{redstyle}
     {3pt}
     {3pt}
     {\color{black}}
     {}
     {\color{red}\bfseries}
     {:}
     {.5em}
     {}
\theoremstyle{redstyle}

\newcommand{\cA}{{\mathcal A}}
\newcommand{\cE}{{\mathcal E}}

\newcommand{\lmin}{\ell_{\min}}
\newcommand{\lmax}{\ell_{\max}}

\newcommand{\labell}[1]{\label{#1}}
\renewcommand{\paragraph}[1]{\vspace*{0.5ex}\noindent {\bf #1}}
\newcommand{\remove}[1]{}

\newcommand{\NAT}{{\mathbb N}}

\usepackage{color}
\newcommand{\dk}[1]{{\color{red}{#1}}}

\newcommand{\comment}[1]{}

\usepackage{changes}
\begin{document}
\definechangesauthor[color=red]{dk}
\definechangesauthor[color=blue, name={Tomasz Jurdziński}]{tj}
\definechangesauthor[color=magenta]{kl}

\title{Online Packet Scheduling\\ under Adversarial Jamming\thanks{%
This work was supported by the Polish National Science Centre grant
DEC-2012/06/M/ST6/00459.
}
}

\author{%
    Tomasz Jurdzinski\footnotemark[2]
    \and
    Dariusz R.~Kowalski\footnotemark[3]
    \and
    Krzysztof Lorys\footnotemark[2]
}

\footnotetext[2]{Institute of Computer Science, University of Wroc{\l}aw, Poland.}

\footnotetext[3]{Department of Computer Science,
            University of Liverpool,
            Liverpool L69 3BX, UK.
            }


\date{}

\maketitle


\begin{abstract}
We consider the problem of scheduling packets of different lengths via a directed communication link prone to jamming errors. Dynamic packet arrivals and errors are modelled by an adversary. We focus on estimating relative throughput of online scheduling algorithms,
that is, the ratio between the throughputs achieved by the algorithm and the best scheduling
for the same arrival and error patterns.
This framework allows more accurate analysis of performance of online scheduling algorithms,
even in worst-case arrival and error scenarios.
We design an online algorithm for scheduling packets of arbitrary lengths,
achieving optimal relative throughput in $(1/3,1/2]$ (the exact value depends on packet lengths).
In other words, for any arrival and jamming patterns, our solution gives throughput
which is no more than $c$ times worse than the best possible scheduling for
these patters, where $c\in [2,3)$ is the inverse of relative throughput.
Another algorithm we design makes use of additional resources in order to achieve relative throughput $1$,
that is, it achieves at least as high throughput as the best schedule without such resources,
for any arrival and jamming patterns.
More precisely, we show that if the algorithm can run with double speed, i.e., with twice higher frequency,
then its relative throughput is $1$. This demonstrates that throughput of the best online scheduling algorithms
scales well with resource augmentation.

\vspace*{.3ex}
\end{abstract}
\ 
\vspace*{-2ex}

\thispagestyle{empty}
\setcounter{page}{0}


\renewcommand{\thefootnote}{\arabic{footnote}}

\newcommand{\alg}{ALG}
\newcommand{\adv}{ADV}
\newcommand{\opt}{OPT}
\newcommand{\off}{OFF}

\newcommand{\mgreedy}{MGreedy}


\section{Introduction}
\label{s:intro}

\paragraph{Motivation.}
%
Achieving high-level reliability in packet scheduling has recently become
more and more important due to substantial increase of the scale of networks and
higher fault-tolerant demands of many incoming applications.
In the era of Internet of Things and nano-devices, it will no longer be possible to
attend devices physically, and therefore the designed protocols must be stable and robust
no matter of failure pattern. Imagine the problem of thousands of malfunctioning
nano-capsules with overflown buffers that need to be somehow removed from the human body,
or the consequences of lack of communication between AVs with humans onboard
or medical devices incorporated into patients bodies,
even if such case might happen with probability less than $1\%$.

\paragraph{Our Approach.}
This paper
%
studies a fundamental problem of online packet scheduling via {\em unreliable} link (also called
a channel),
when transmitted packets may be interrupted by {\em unrestricted} jamming errors.
This problem was recently introduced in~\cite{Anta-SIROCCO-2013} and analyzed for
two different packet lengths.
Packets arrive dynamically to one end of the link, called a sender, and need
to be transmitted in full, i.e., without any in-between jamming error, to the other end (called a receiver).
Jamming errors are immediately discovered by the sender.
We analyze all possible scenarios, including worst case ones,
which we model as a conceptually adversary who controls both packet arrivals and channel jamming.
The adversary is unrestricted, in the sense that she may generate {\em any} arrival and error
pattern.
The main objective of the online scheduling protocol is to achieve as high throughput
as possible under current scenario.
In particular, we consider the measure called relative throughput,
which is a long-term form of competitive ratio between the throughput achieved by
the online algorithm and the one reached by optimum offline scheduling solution (i.e., under
the knowledge of adversarial arrivals and errors).

\comment{

[[[PONIZSZY paragraph do usuniecia, fragment o feedback wrzuce do related work]]]

\noindent\emph{Feedback mechanisms:}
Then, moving to the online problem requires detecting 
the packets received with errors, in order to retransmit them.
The usual mechanism~\cite{ARQ_CRC}, which we call {\em deferred feedback}, detects and notifies the sender that a packet has suffered an error after the whole packet has been received by the receiver.
It can be shown that,
even when the packet arrivals are stochastic and packets have the same length, no online scheduling algorithm with deferred feedback can
be competitive with respect to the offline one.
Hence, we center our study in a second mechanism, which we call {\em instantaneous feedback}. It detects and notifies the sender of an error the moment this error occurs.
This mechanism can be thought of as an abstraction of the emerging
Continuous Error Detection (CED) framework~\cite{CED} that uses arithmetic coding to provide continuous error detection.
The difference between deferred and instantaneous feedback is drastic, since for the instantaneous feedback mechanism, and for packets of the same length, it is easy to obtain optimal relative throughput of 1, even in the case of adversarial arrivals.
However, the problem becomes substantially more challenging in the case of non-uniform packet lengths. Hence, we analyze the problem for the case of packets with two different lengths, $\lmin$ and $\lmax$, where $\lmin<\lmax$. \vspace{.2em}

}

\newcommand{\SL}{SL}
\newcommand{\LL}{LL}
\newcommand{\Ralg}{RAlg}
\newcommand{\CRalg}{CRAlg}
\newcommand{\Alg}{\alg}
\newcommand{\rhoflr}{\overline{\rho}}

\paragraph{Our Contribution.}
We design a deterministic online scheduling algorithm achieving optimal relative throughput
for an arbitrary number $k$ of packet lengths $\lmin=\ell_1<\ell_2<\ldots<\ell_k=\lmax$
(Section~\ref{sec:no-speedup}).
We first show a simpler version of the algorithm,
for the case when packet lengths are pairwise divisible (any larger is divisible by any smaller),
in order to demonstrate high-level ideas and analysis leading to related throughput $1/2$.
We then extend the protocol so that it does not need to rely on such limitation about divisibility,
and achieves the relative throughput
$\min_{1\le j < i \le k}\left\{\frac{\lfloor\rho_{i,j}\rfloor}{\lfloor\rho_{i,j}\rfloor+\rho_{i,j}} \right\}$,
where $\rho_{i,j}=\ell_i/\ell_j$ is the ratio between the $i$-th and the $j$-th packet length.
Note that this general formula for relative throughput is
in the range $(1/3,1/2]$,
and it reaches $1/2$ if and only if the pairwise divisibility condition holds.

Unfortunately, the designed protocol does not achieve relative throughput $1$ if the speed-up $2$
is applied (it can be easily checked that the relative throughput is at most $2/3$ in such case),
which implies that it is not well-scalable with resource augmentation.\footnote{%
Note that the considered speed-up $2$ is chosen because we claim linear scalability
of relative throughput with the increase of speed-up, that is, starting from level $1/2$ with no speed-up
we expect the relative throughput to reach value $1$ for speed-up $2$.}
Therefore we design another deterministic online protocol to optimize relative throughput for speedup $2$ (Section~\ref{sec:speedup}).
It is a generalisation of the preamble
protocols proposed in~\cite{Anta-SIROCCO-2013} and~\cite{Anta-FCT-2013} in the case of two
packet lengths.

More details can be found in the full draft of the paper~\cite{arxiv-full}.


\comment{

(Section~\ref{sec:adversarial}), that an online algorithm with instantaneous feedback can achieve at most almost half the relative throughput with respect to the offline one.
It can also be shown
that two basic scheduling policies,
giving
priority either to short {\em($\SL$ -- Shortest Length)} or long
{\em($\LL$ -- Longest Length)} packets, are not efficient under adversarial errors.
Therefore,
we devise a new algorithm, called \Ralg, and show that it achieves the optimal online relative throughput.
Our algorithm, transmits a ``sufficiently'' large number of short packets while making sure that long packets are transmitted from time~to~time.

\begin{table}
\begin{footnotesize}
\begin{center}
	\begin{tabular}{| c | c | c | c |}
	\hline
	 Arrivals & Feedback & Upper Bound & Lower Bound \\ \hline \hline
	 & Deferred  & $0$ & $0$ \\ \cline{2-4}
	 & & & \\	
	Adversarial &  Instantaneous & $T_{\Alg} \leq \rhoflr/(\rho + \rhoflr)$ & $T_{SL-Pr} \geq \rhoflr/(\rho + \rhoflr)$ \\
	  &  &  & \\
	 &  &  $T_{LL} = 0$, $T_{SL} \leq 1/(\rho + 1)$ & \\ \hline \hline
	
	 & Deferred & $0$ & $0$ \\ \cline{2-4}
	 & & & \\	
	Stochastic & Instantaneous& $T_{\Alg} \leq \rhoflr/\rho$ &
	$T_{CSL-Pr} \geq \rhoflr/(\rho + \rhoflr)$, if $\lambda p \lmin \le \rhoflr/ (2\rho)$ \\
	 & & $T_{\Alg} \leq \max\left\{ \lambda p\lmin, \rhoflr/(\rho + \rhoflr)\right\}$, if $p<q$ & $T_{CSL-Pr} \geq \min\left\{ \lambda p\lmin, \rhoflr/\rho \right\}$, otherwise \\
	 & & $T_{LL} = 0$, $T_{SL} \leq 1 / (\rho + 1) $  & \\ \hline 	 	
	\end{tabular}
	\caption{Summary of results presented. The results for deferred feedback are for one packet length, while the results for instantaneous feedback are for 2 packet lengths $\lmin$ and $\lmax$. Note that $\rho = \lmax / \lmin$, $\rhoflr = \lfloor\rho\rfloor$, $\lambda p$ is the arrival rate of $\lmin$ packets, and $p$ and $q=1-p$ are the proportions of $\lmin$ and $\lmax$ packets, respectively.
	}
	\label{t:1}
\end{center}
\end{footnotesize}
\end{table}

}




\paragraph{Previous and related work.}
Packet scheduling~\cite{packet_scheduling} is one of the most fundamental problems in computer networks.
A realistic approach involves {\em online} scheduling~\cite{Peleg1992,online_scheduling},
and therefore a {\em competitive analysis}~\cite{aspnes,tarjan} is often used to evaluate the performance of
proposed solutions.
%
Online scheduling was considered in a number of models;
for more information the reader is referred to~\cite{scheduling} and~\cite{online_scheduling}.


The framework considered in this work was recently introduced in~\cite{Anta-SIROCCO-2013}.
The authors showed that general offline version of this problem, in which the scheduling algorithm knows a priori when errors will occur, is NP-hard.
They also considered algorithms and upper limitations for relative
throughput in case of {\em two} packet lengths.
In particular, they proved that relative throughput of {\em any} online scheduling protocol cannot
be bigger than $\rhoflr/(\rho + \rhoflr)$, where $\rho$ is the ratio between the bigger and the smaller
packet length and $\rhoflr=\lfloor\rho\rfloor$.
(Note that the upper bound becomes $1/2$ if the bigger packet length is a multiplicity of the smaller
packet length.)
This upper bound can be achieved by a protocol scheduling a specific preamble of shorter packets
followed by the Longest\_First rule after every error, but cannot be reached by simpler protocols
such as Longest\_First itself or Shortest\_First (in fact, the relative throughputs of the latter
protocols are far worse than $1/2$: $0$ and $1/(1+\rho)$, respectively, and thus they are not very reliable).
Therefore, it remained open whether there is an online scheduling protocol reaching
the relative throughput of (roughly) $1/2$ for {\em arbitrary} number of packet lengths;
we answer this question in affirmative in this work.
Moreover, as also shown in~\cite{Anta-SIROCCO-2013}, randomization does not help,
which motivates our study of deterministic algorithms.

In~\cite{Anta-FCT-2013}, the authors studied buffer sizes of online scheduling protocols on
error-prone channel. Unlike the relative throughput measure, in order to be positively competitive with
the best scheduling algorithms with respect to the buffer sizes, additional resources need to be given
to the online protocol, i.e., speed-up (higher frequency). This form of resource allocation appeared to be
efficient: for some speed-up smaller than $2$ there is a deterministic online scheduling algorithm
having roughly the same queue sizes as any other scheduling algorithm running without speed-up.
That work motivated us to consider resource augmentation technique,
in the form of using some speed-up (higher frequency),
to reach at least the same throughput as the best scheduler without speed-up
for {\em every} execution.

Wireless packet scheduling was also considered in models with physical constraints included,
such as radio networks or SINR.
Anantharamu et al.~\cite{ACKR-INFOCOM-2010} considered packet scheduling on a multiple
access channel with signal interference, under a restricted adversarial patterns of packet arrivals
and channel jamming.
Kesselheim~\cite{kesselheim} considered packet scheduling problem in
the SINR model,
for both adversarial and stochastic arrivals, but with no errors.
Both papers studied stronger objective measure: maximum time from packet arrival to successful delivery.
%
In this line of research, the most relevant direction was taken by
Richa et al.~\cite{jamming} 
who analyzed competitive throughput of randomized scheduling protocols
on multiple access channels with signal interference
against {\em adaptive, but still restricted, adversarial jamming}.
Therefore, the results obtained in this line of work cannot be directly comparable
with ours, mainly because of assuming restricted arrival and jamming patterns.


Andrews and Zhang~\cite{HSWD} studied buffer stability (i.e., bounded buffers property)
of online packet scheduling on
a wireless channel, where both the channel conditions and the data
arrivals are controlled by an adversary.
They also assumed bounded adversary, as otherwise stability could not be reached.

Our framework could also be applied to other types of channel errors, as long as the feedback
is immediately delivered to the sender, e.g.,
to the emerging Continuous Error Detection (CED) framework~\cite{CED} that uses arithmetic coding to provide continuous error detection.

\section{Model}


We consider a uni-directional point-to-point link in which one end point, called a {\em sender},
transmits packets to the other end point, called a {\em receiver}.
The sender is equipped with unlimited buffer, in which the arriving packets are queued.
Packets may be of different lengths, and may arrive at any time; we assume that time
is continuous, and scheduling algorithm have access to packets as soon as they arrive.
There are $k\ge 2$ different packet lengths, denoted by
$\lmin=\ell_1<\ell_2<\ldots<\ell_k=\lmax$.
For simplicity, we will use the names ``$\ell_i$-packets'' and ``packets $\ell_i$'' 
for packets of length $\ell_i$, for any $1\le i \le k$.
For clarity of presentation, we assume in some parts of the paper 
that $\ell_i/\ell_j$ is an integer for any $1\le j < i \le k$ (so called pairwise divisibility property).
We denote $\rho = \lmax/\lmin$.
We assume that all packets are transmitted at the same bit rate, hence the transmission time is proportional to the packet's length.
The link is prone to jamming errors, that is, transmitted packets might be corrupted at any time point.

\paragraph{\bf Arrival models.}
We consider adversarial packet arrivals:
the packets' arrival time and length are governed by an adversary.
We define an adversarial arrival pattern as a collection of packet arrivals caused by the adversary.

\paragraph{\bf Link jamming errors.}
We consider adversarial model of jamming errors, in which the adversary decides at which time to cause a
jamming error
on the link. The error at time $t$ implies that any packet being transmitted at time $t$ is broken,
and the information about it is immediately delivered to the sender so that it breaks the current transmission and could schedule another packet (or re-schedule the one that was just broken).
A corrupted packet transmission is unsuccessful, in the sense that it is not received by the receiver
and it needs to be retransmitted in full (not necessarily right after the error --- scheduling algorithm
may decide to postpone it and transmit another packet instead); 
otherwise it is understood as not (successfully) transmitted.
We assume that scheduling algorithms do not voluntarily stop transmitting packets before
the end of the transmission, unless they get feedback about jamming error.
An adversarial error pattern is defined as a collection of error events on the link caused by the adversary.

Adversarial models are typically used to argue about the algorithm's behavior in 
any possible scenario, in particular, in the worst-case ones.
%

\paragraph{\bf Efficiency metric: \em Relative throughput.}
We would like to measure throughput of the communication link.
However, due to adversarial errors, the real link capacity may vary in time,
and moreover, due to adversarial packet arrivals, the stream of packets may not be regular
or saturated. Therefore, following~\cite{Anta-SIROCCO-2013}, we pursue a long-term competitive analysis.
Let $\cA$ be an arrival pattern and $\cE$ an error pattern.
For a given deterministic algorithm {\alg},
let $L_{\alg}(\cA,\cE,t)$ be the total length of all the successfully transferred (i.e., non-corrupted) packets by time $t$ under arrival pattern $\cA$ and error pattern $\cE$.
Let {\opt} be the offline optimal algorithm that knows the exact arrival and error patterns, as well as the
online algorithm, before the start of the execution.
We assume that {\opt} devises an optimal schedule that 
minimises the asymptotic ratio (i.e., with time growing to infinity) between the total length of packet transmitted by the online algorithm and the total length of packet transmitted by itself. 
%

%
We require that any pair of patterns $\cA,\cE$
occurring in an execution must allow non-trivial communication, i.e., the value of $L_{\opt}(\cA,\cE,t)$ in the execution is unbounded with $t$ going to infinity.

For arrival pattern $\cA$, adversarial error pattern $\cE$ and time $t$, we define the \emph{relative throughput $T_{\alg}(\cA,\cE,t)$ of a deterministic algorithm $\alg$ by time $t$} as:
\[
T_{\alg}(\cA,\cE,t) = \frac{L_{\alg}(\cA,\cE,t)}{L_{\opt}(\cA,\cE,t)}
\ .
\]
For completeness, $T_{\alg}(\cA,\cE,t)$ equals 1 if $L_{\alg}(\cA,\cE,t)=L_{\opt}(\cA,\cE,t)=0$.

We define the {\em relative throughput} of $\alg$
in the adversarial arrival model as:
\begin{equation}
\label{eq:rel-thr}
T_{\alg} = \inf_{\cA , \cE} \lim_{t \rightarrow \infty} T_{\alg}(\cA,\cE,t)
\ .
\end{equation}


In the analysis of lower and upper bound on relative throughput, we usually focus on comparison
of the number of successful transmissions of packets, weighted by packet lengths, for
periods after {\em sufficiently large} time $t$.
This is because the performances of online and optimal algorithms in a fixed prefix of time
are negligible from perspective of the definition of relative throughput given in Equation~(\ref{eq:rel-thr}).

\paragraph{Resource augmentation --- speed-up.}
In the second part of the paper, in Section~\ref{sec:speedup}, we consider resource augmentation
technique. This technique was recently applied to fault-tolerant scheduling in~\cite{Anta-FCT-2013}
in the context of buffer stability metric. 
In particular, we compare the throughput of a given online algorithm under the assumption
that this algorithm is run with a certain speed-up $s>1$, with the throughput of the best scheduling algorithm 
run without any speed-up. From technical perspective, computing of the relative throughput
under speed-up $s>1$ follows the same definitions as given above, with the only difference that the value
of $L_{\alg}(\cA,\cE,t)$ is calculated under assumption that $\alg$ transmits packets $s$ times faster.

\comment{

\dk{[[[Ponizsze skopiowane z pracy z sirocco - nie wiem czy potrzebne]]]}
Finally, we consider {\em work conserving} online scheduling algorithms, in the following sense:
as long as there are pending packets, the sender does not cease to schedule packets.
Note that it does not make any difference whether one assumes that offline algorithms are work-conserving or not, since their throughput is the same in both cases (a work conserving offline algorithm always transmits, but stops the ongoing transmission as soon as an error occurs and then continues with the next packet). Hence for simplicity we do not assume offline algorithms to be work conserving.

}

\newcommand{\greedy}{Greedy}
\newcommand{\group}{\text{Transmit-group}}

\section{Packet Scheduling for $k$ packet lengths}\label{sec:no-speedup}

\comment{
In this section we present an algorithm which achieves optimal
relative throughput. First, a simpler version is presented which
achieves the optimal relative throughput $1/2$ provided $\ell_i/\ell_{i-1}\in\NAT$
for each $i\in[2,k]$.
Then, we generalize this solution to the case that the lengths of packets
are arbitrary and the relative throughput is given by an expression depending
on the lengths of packets. This result achieves the optimal relative throughput
as well.

\subsection{Packets lengths with divisibility property}
}

In this section we present
online algorithm,
which is optimal for any number of packet lengths $k\geq 2$.
First, for the ease of presentation, we present algorithm {\greedy} under assumption that
$\ell_i/\ell_{i-1}\in\NAT$ for $1<i<k$.
Later, in Section~\ref{s:mgreedy}, we
show how to remove this assumption by modifying algorithm {\greedy};
the resulted algorithm is called {\mgreedy}.

The main idea behind our algorithms is to keep transmitting as many
short packets as possible (shortest-first strategy), subject to some balancing constraints.
Observe that it is difficult for any offline algorithm {\off} to get advantage over any
online algorithm {\alg} when {\alg}
sends small packets. Thus, preference for small packets ensures that {\alg} can be competitive
against {\off}, as long as it has short packets. However, if {\off} transmits
large packets during transmission of small packets by {\alg}, it can afterwards
transmit small packets when {\alg} does not have any of them in its queue.
Simultaneously, when {\off} is transmitting small packets, {\adv} can generate errors preventing {\alg} from
successful transmission of large packets. Despite this disadvantage
of a greedy approach, we show that an appropriate implementation of
this strategy, using some balancing constraints,
provides an optimal solution with respect to relative throughput,
and thus against any optimal way of scheduling under occurring arrival and failure patterns.

Our specific modification of the greedy shortest-first strategy is based on sending
packets in groups, which altogether balance the length of the next larger packet.
We explain it in detail first for two types of packet lengths: $\lmin$ and $\lmax$.
If there are at least $\rho=\lmax/\lmin$ small packets
in the queue, the algorithm builds a {\em group} which consists of $\rho$ of them
and keeps sending them until all of them are transmitted successfully.
If there are less than $\rho$ small packets in the queue at the moment
when a transmission of a group is finished, a large packet is transmitted.
However, whenever there are at least $\rho$ small packets, the group of small
packets is formed, independently of the fact whether a transmission of a large
packet(s) is successful or not. This idea is then recursively applied for the case
when there are $k>2$ types of packets. A pseudo-code of our greedy
algorithm is presented as Algorithm~\ref{alg:greedy}, with its recursive
subroutine given as Algorithm~\ref{alg:group}.

\begin{algorithm}[h]
	\caption{{\greedy}}
	\label{alg:greedy}
	\begin{algorithmic}[1]
    \Loop
        \While{$\sum_{i=1}^k\ell_i n_i < \ell_k$} Stay {\em idle}
        \EndWhile
        \State Transmit-group$(k)$
    \EndLoop
    \end{algorithmic}
\end{algorithm}

\begin{algorithm}[h]
	\caption{Transmit-group$(j)$}
	\label{alg:group}
	\begin{algorithmic}[1]

    \Loop
        \If{$\sum_{i=1}^{j-1}\ell_i n_i\geq \ell_j$}
            \For{$a=1$ \textbf{to} $\ell_j/\ell_{j-1}$}
                Transmit-group$(j-1)$
            \EndFor
            \Return
        \EndIf
            \State Transmit $\ell_j$; If the transmission is successful: \Return 

    \EndLoop
    \end{algorithmic}
\end{algorithm}


In the pseudo-codes, $n_i$ denotes the number of packets $\ell_i$ which are
currently (at the moment) waiting in the queue for transmission.

\paragraph{Performance analysis of algorithm {\greedy}.}

For the sake of analysis of algorithm {\greedy}, we introduce some
new notations. First, let us assume that an arrival pattern and an injection
pattern are chosen arbitrarily and are fixed, so we could omit them from formulas in the further analysis.
For an algorithm $A$, let $q_A(i,t)$ denote the sum of lengths
of $\ell_i$-packets in the queue of $A$ at the moment $t$.
That is, $q_A(i,t)=n_i\cdot \ell_i$ for a fixed time $t$.
Moreover, let
$q_A(<i,t)=\sum_{j<i}q_A(j,t)$ and we define $q_A(\leq i, t)$ analogously.
Let $L_A(i,t)$ denote the length of packets $\ell_i$ successfully transmitted
by time $t$. For a time period $\tau=[t_1,t_2]$, let $L_A(i,\tau)=L_A(i,t_2)-L_A(i,t_1)$.
That is, $L_A(i,\tau)$ denotes the number of $\ell_i$-packets successfully transmitted
in the interval $\tau$.
The notions $L_A(<i,t)$, $L_A(\leq i,t)$, $L_A(<i,\tau)$, and $L_A(\leq i,\tau)$
for time $t$ and time interval $\tau$ are defined analogously to
$q_A(<i,t)$, $q_A(\leq i,t)$, $q_A(<i,\tau)$ and $q_A(\leq i,\tau)$.
We also use the above introduced notations without the first argument, i.e.,
$q_A(t)$, $q_A(\tau)$, $L_A(t)$, and $L_A(\tau)$, which are shorthands for
$q_A(\leq k, t)$, $q_A(\leq k, \tau)$, $L_A(\leq k, t)$ and $L_A(\leq k, \tau)$,
respectively.

An algorithm $A$ is {\em busy} at time $t$ if it is transmitting a packet at $t$,
it has just finished a successful transmission, or its transmission is jammed by
an error at $t$. Otherwise $A$ is {\em idle} at $t$.

Our goal is to compare progress in sending packets of our algorithm {\greedy} and
an algorithm achieving the optimal throughput, denoted as {\off}. We say that
an algorithm $A$ is {\em $m$-busy} in a time period $\tau=[t_1,t_2]$ if the following conditions
are satisfied:
\begin{enumerate}
\item
$A$ is busy at each time $t\in\tau$;
\item
$A$ does not transmit packets $\ell_i$ for $i>m$ during $\tau$;
\item
$q_A(i,t_1)\geq q_{\off}(i,t_1)$ for each $i\in[m]$.
(That is, at time $t_1$ $A$ has no less
 packets of length $\ell_i$ in its queue than {\off}, for each $i\leq m$.)
\end{enumerate}

Now, we prove technical results regarding periods in which {\greedy} is $m$-busy for some $m\in[k]$.
These lemmas eventually lead to the proof of the fact that relative throughput
of {\greedy} is $1/2$ (provided $\ell_i/\ell_{i-1}\in\NAT$ for $i\in[2,k]$), which is optimal.
First, we make an observation that, if {\greedy} does not use packets longer than
$\ell_m$ for $m\in[k]$, then the total length of packets transmitted by {\greedy} is
at least as large as the total length of packets of length at least $\ell_m$ transmitted by {\off}.

\begin{lemma}\label{lem:k1}
Assume that {\greedy} is $m$-busy in a time period $\tau$, $m\leq k$.
Then, $L_{\greedy}(\tau)\geq L_{\off}(\geq m, \tau)-\ell_k$.
\end{lemma}

\begin{proof}
Consider any packet $\ell_i$ for $i\geq m$ successfully transmitted
by {\off} in the period $[t,t+\ell_i]\subseteq\tau$. According to the
assumptions, {\greedy}
does not send any packets $\ell_j$ for $j>m$ in $\tau$ and it is busy
at each time $t'\in \tau$. Therefore, {\greedy} finishes transmissions
of $\ell_i/\ell_m$ groups of packets of length $\ell_m$ in the period
$[t,t+\ell_i]$. Thus, by assigning each group $G$ of packets of length $\ell_m$
transmitted by {\greedy} to a packet transmitted by {\off} at the moment
when the last packet of $G$ is finished, we obtain the result of the
lemma. The ``$-\ell_k$'' reduction in the formula is needed in order to take into account the packet (if any)
which {\off} started transmitting before $\tau$ and finished in $\tau$.
\end{proof}
Next, we formulate a relationship between the length of packets transmitted
by {\greedy} and {\off} up to the moment when {\greedy} is transmitting the
longest packet used by itself in the computation.
\begin{lemma}\label{lem:k2}
Assume that {\greedy} is $m$-busy in a time period $\tau=[t_1,t_2]$, $m\leq k$.
Let $t\in\tau$ be any time at which {\greedy} starts transmitting
$\ell_m$. Then,
$$2L_{\greedy}([t_1,t])\geq L_{\off}([t_1,t])+q_{\off}(<m,t)-\ell_m-\ell_k.$$
\end{lemma}
\begin{proof}
The idea is that each packet $p$ successfully transmitted by {\greedy}
is associated to:
\begin{enumerate}
\item[(a)]
transmission of $p$ by {\off};
\item[(b)]
transmission of a packet $\ell_i$ for $i\geq m$ by {\off} which lasted while the
group $G$ of length $\ell_i$ containing $p$ was finished by {\greedy}.
\end{enumerate}
This association guarantees that:
\begin{itemize}
\item
each packet $\ell_i$, for $i>m$, transmitted by {\off} has an association (of type (b)) with a group
of packets of lengths $\ell_i$ transmitted by {\greedy};
\item
each packet $\ell_i$, for $i\leq m$, transmitted by {\off} has an association (of type (a)) with
its transmission by {\greedy}, provided {\greedy} transmitted this packet successfully as well.
\end{itemize}
On the other hand, each packet $p$ successfully transmitted by {\greedy} corresponds to
successful transmissions of {\off} with length at most twice the length of $p$.
However, packets which are successfully transmitted by {\greedy} in $[t_1,t]$ and
are not transmitted by {\off} in $[t_1,t]$ ``pay'' for transmissions of {\off} only once, i.e., are associated
to a packet $\ell_i$, for $i> m$, transmitted by {\off} but not to transmissions of themselves by {\off}.
As {\greedy} tries transmitting $\ell_m$ at time $t$ only in the case when $q_{\greedy}(<m,t)<\ell_m$,
the claimed result holds.

The ``$-\ell_k$'' reduction in the formula is needed in order to take into account the packet
which {\off} started transmitting before $\tau$ and finished in $\tau$.
\end{proof}
Using previous lemmas, we prove by induction a relationship between $L_{\greedy}(\tau)$
and $L_{\off}(\tau)$ for periods $\tau$ which are $m$-busy for {\greedy}, where $m\in[k]$.

\begin{lemma}\label{lem:k3}
Assume that {\greedy} is $m$-busy in a time period $\tau$, for $m\leq k$.
Then,
$$2L_{\greedy}(\tau)\geq L_{\off}(\tau)-f_m$$ 
where $f_m$ satisfies the relationships
$f_1= \ell_k$ and $f_{i+1}=f_i+3\ell_{i+1}+\ell_i+\ell_k$ for $i\in[1,k-1]$.
\end{lemma}
\begin{proof}
The proof goes by induction with respect to $m$. For $m=1$ the result is
an immediate consequence of Lemma~\ref{lem:k1}.

For the inductive step, assume that the result holds for some $m<k$. We will show the correctness of the
result for the case when the longest packet sent by {\greedy} in $\tau$ is $\ell_{m+1}$.
We split $\tau$ into three subintervals:
\begin{itemize}
\item
$\tau_1$ from the beginning of $\tau$ to time $t$ at which {\greedy} starts (an attempt to) transmitting
$\ell_{m+1}$ for the last time during $\tau$;
\item
$\tau_2$ from $t$ to $t'\in\tau$ such that either {\greedy} finishes a successful transmission of $\ell_{m+1}$
at $t'$ or it gives up scheduling packets $\ell_{m+1}$ at $t'$ (since it has enough shorter packets in the queue
at $t'$ to cover the length $\ell_{m+1}$);
\item
$\tau_3$ from $t'$ to the end of $\tau$.
\end{itemize}
Lemma~\ref{lem:k2} implies that
$$2L_{\greedy}(\tau_1)\geq L_{\off}(\tau_1)+q_{\off}(<m+1,t)-\ell_{m+1}-\ell_k,$$
where $t$ is the moment when $\tau_1$ ends.
Consider {\off'} which acts in $\tau_2$ and $\tau_3$ as {\off}, however: it starts
$\tau_2$ without packets of length $\ell_i$ for each $i< m+1$, and it stays idle each time {\off} is transmitting
a packet which was in its queue at the beginning of $\tau_2$ and therefore it was not in the queue of {\off'}.

Note that {\greedy} finishes an attempt to transmit a packet $\ell_{m+1}$ not later than at
the moment when it successfully transmits packet $\ell_{m+1}$ or new packets shorter than
$\ell_{m+1}$ of overall length at least $\ell_{m+1}$ are inserted in the queue and error
occurs. Observe also that {\off'} starts time interval $\tau_2$ with an empty queue and it cannot finish transmitting
a packet $\ell_i$, for $i>m+1$, in $\tau_2$ (if a time period of length $\ell_{m+1}$ without error
occurs, $\tau_2$ is finished by its definition).
Thus,
$$L_{\greedy}(\tau_2)\geq L_{\off'}(\tau_2)-2\ell_{m+1} \ ,$$
since new packets inserted during $\tau_2$ and transmitted by {\off'} have length $<\ell_{m+1}$:
\begin{itemize}
\item
packets of length smaller than $\ell_{m+1}$ are inserted until the beginning of the last attempt
to send $\ell_{m+1}$ by {\greedy} in~$\tau_2$;
\item
packets of length at most $\ell_{m+1}$ are successfully transmitted by {\off'} during
the last attempt to send $\ell_{m+1}$ by {\greedy}, since this attempt takes a time
period of length at most $\ell_{m+1}$.
\end{itemize}
At the beginning of $\tau_3$, it holds that $q_{\greedy}(i,t')\geq q_{\off'}(i,t')$
for $i\leq m$, since {\greedy} attempted only transmitting one copy of $\ell_{m+1}$
during $\tau_2$.

Therefore, the inductive hypothesis apply to the period $\tau_3$ for the largest
packet $\ell_m$ and {\off'} in place of {\off}:
$$2L_{\greedy}(\tau_3)\geq L_{\off'}(\tau_3)-f_m \ .$$ 
Next, recall that {\off'} differs from {\off} only such that it does
not transmit packets of lengths smaller than $\ell_{m+1}$, which appear in the queue of {\off}
at time $t$, of total length $q_{\off'}(<m+1,t)$. Thus,
$$L_{\off'}(\tau_2\cup\tau_3)\geq L_{\off}(\tau_2\cup\tau_3)-q_{\off}(<m+1,t) \ .$$
All these inequalities summed up and combined with the fact that
$L_{\greedy}(\tau)=L_{\greedy}(\tau_1)+L_{\greedy}(\tau_2\cup\tau_3)$
yield:
$$\begin{array}{rcl}
2L_{\greedy}(\tau)&\geq& L_{\off}(\tau_1)+(q_{\off}(<m+1,t)-\ell_{m+1}\\
&& -\ell_k)+ L_{\off'}(\tau_2\cup\tau_3)-f_m-2\ell_{m+1}\\
&\geq& L_{\off}(\tau)-f_m-3\ell_{m+1}-\ell_m-\ell_k\\
&\geq& L_{\off}(\tau)-f_{m+1}
\ .
\end{array}$$
\end{proof}


\begin{theorem}\label{th:kdivlower}
The relative throughput of {\greedy} is equal to $1/2$, provided $l_i/l_{i-1}\in\NAT$
for each $i\in[2,k]$.
\end{theorem}
\begin{proof}
Lemma~\ref{lem:k3} implies that the relative throughput gets arbitrarily close to $1/2$
on sufficiently long time intervals in which {\greedy} is busy and {\off} starts with
the queue containing smaller or equal number of packets of each size.
On the other hand, {\greedy} gets idle only in the case when its queue is almost empty,
i.e., contains packets of length $<\ell_k$, which means that its relative throughput is
arbitrarily close to $1$ in such a moment $t$, provided $t$ is large enough.

The theorem holds by combining these two observations:
\begin{itemize}
\item
It holds at the moments when {\greedy} is idle.
\item
We can assume that {\off} has empty queues at the moments
when {\greedy} starts being idle --- this assumption does not improve relative
throughput of {\greedy} (we allow {\off} to transmit all packets from
its queue immediately in a period of length $0$). Thus, the assumption of Lemma~\ref{lem:k3} is satisfied
when {\greedy} starts transmitting after being idle. Therefore, the relative throughput
of {\greedy} gets arbitrary close to $\frac12$ in sufficiently large periods in
which {\greedy} is not idle.
\end{itemize}
\end{proof}

\begin{corollary}
The algorithm {\greedy} achieves optimal relative throughput for packets'
lengths $\ell_1<\ldots<\ell_k$ such that $\ell_i/\ell_{i-1}\in\NAT$
for each $i\in[2,k]$.
\end{corollary}
\begin{proof}
It is shown in \cite{Anta-SIROCCO-2013} that relative throughput of any online
algorithm for two types of packets is at most
$$\frac{\lceil \ell_2/\ell_1\rceil}{\lceil \ell_2/\ell_1\rceil+\ell_2/\ell_1}$$
which is equal to $\frac12$ when $\ell_2/\ell_1\in\NAT$. As an adversary can
decide to schedule merely two types of packets among available $k$ types,
Theorem~\ref{th:kdivlower} implies optimality of relative throughput
of {\greedy}.
\end{proof}

\subsection{Arbitrary lengths of packets}
\label{s:mgreedy}

In this section we discuss an application of the ideas behind the algorithm
{\greedy} to the general case, i.e., when the condition
$\ell_i/\ell_{i-1}\in\NAT$ is not satisfied.
Let $\rho_{i,j}=\ell_i/\ell_{j}$.
A natural generalization of {\greedy} is
that, instead of $\rho_i$ groups of packets of length $\ell_{i-1}$
on the $i$-th level of recursion, we choose $\lfloor \rho_{i,i-1}\rfloor$
groups of packets of length (as close as possible to) $\ell_{i-1}$ in order to ``cover'' $\ell_i$.
If the length of a group of packets on the $i$-th level of recursion
is not larger than $\ell_{i+1}$, we can apply the ideas of ``covering'' packets transmitted
by {\off} using groups of packets transmitted by {\greedy}.
If $k=2$, this approach gives an algorithm with relative
throughput
$\frac{\lfloor \rho_{2,1}\rfloor}{\lfloor \rho_{2,1}\rfloor+\rho_{2,1}}$,
which is optimal due to~\cite{Anta-SIROCCO-2013}.
This naturally generalizes to the following result.

\begin{theorem}\labell{th:nospeed:upper}
The relative throughput of any online scheduling algorithm is at most
$$\min_{1\le j < i \le k}\left\{\frac{\lfloor\rho_{i,j}\rfloor}{\lfloor\rho_{i,j}\rfloor+\rho_{i,j}} \right\} \ .$$
\end{theorem}

\begin{proof}
This result easily follows from Theorem~1 in \cite{Anta-SIROCCO-2013}. Indeed, if the adversary
schedules merely packets $\ell_i$ and $\ell_{j}$, for $i,j$ minimizing the expression
$\frac{\lfloor\rho_{i,j}\rfloor}{\lfloor\rho_{i,j}\rfloor+\rho_{i,j}}$, the strategy described in the
proof of Theorem~1 in \cite{Anta-SIROCCO-2013} gives our result.
\end{proof}
However, for $k>2$, the additional advantage of {\off} over {\greedy} following from
rounding on various levels of recursion can accumulate.
In order to limit this effect, instead of transmitting $\lfloor \ell_i/\ell_{i-1}\rfloor$ groups of packets
on the level $i-1$, we keep sending groups on the level $i-1$
as long as the sum of lengths of packets from the transmitted groups is not larger than
$\ell_i-\ell_{i-1}$.
This gives the following technical result.
(For simplifying the arguments in the remaining part of the analysis, let us denote $\rho_{i,i-1}=\ell_i/\ell_{i-1}$
by simply $\rho_i$, for $i\in [2,k]$.)

\begin{lemma}\labell{lem:nospeed:aux}
Consider such a modification of {\greedy} that $\group(j)$ keeps calling $\group(j-1)$,
for $j>1$, as long as the total length of transmitted packets in the current execution of $\group(j)$
is at most $\ell_j-\ell_{j-1}$.
The relative throughput of this algorithm
is at least
$\min_{i\in[2,k]}\left\{\frac{\rho_i-1}{2\rho_i-1} \right\} \ .$
\end{lemma}

\begin{proof}
In Lemmas~\ref{lem:k1}, \ref{lem:k2} and \ref{lem:k3}, we repeatedly use an argument
that, if {\greedy} does not use packets of length $\ell_i$ for $i>m$, then each
such packet transmitted by {\off} corresponds to a group of (shorter) packets transmitted
by {\greedy} of total length $\ell_i$. This observation can be preserved for the modified
{\greedy} algorithm with a relaxation that a packet $\ell_i$ transmitted by {\off}
corresponds  to a group of packets transmitted
by {\greedy} of length at least $\ell_i-\ell_{i-1}$.
%
%
This relaxation translates inequalities from Lemmas~\ref{lem:k1}, \ref{lem:k2} and \ref{lem:k3} to:
\begin{equation*}
\begin{array}{rcl}
\frac{\rho_m}{\rho_m-1}\cdot L_{\greedy}(\tau)&\geq& L_{\off}(\geq m, \tau)-\ell_k\\
(1+\frac{\rho_m}{\rho_m-1})\cdot L_{\greedy}([t_1,t])&\geq & L_{\off}([t_1,t])+\\
&&q_{\off}(<m,t)-\ell_m-\ell_k\\
(1+\frac{\rho_m}{\rho_m-1})\cdot L_{\greedy}(\tau)&\geq & L_{\off}(\tau)-f_m
\end{array}
\end{equation*}
If we apply the above inequalities instead of those from Lemmas~\ref{lem:k1}, \ref{lem:k2} and \ref{lem:k3}
in the proof of Theorem~\ref{th:kdivlower}, we obtain the result claimed here. 
\end{proof}

However, as a group of packets transmitted by {\greedy} ``covering'' $\ell_i$ transmitted by {\off}
may contain packets of various lengths, the relative throughput of the solution from
Lemma~\ref{lem:nospeed:aux} is
difficult to compare with the upper bound from Theorem~\ref{th:nospeed:upper}.
In order to tackle this issue, we introduce yet another modification to the algorithm.

The main goal of this modification is to ensure that {\greedy} is transmitting packets of the same
length for long periods of time and it changes to other length only if it is necessary. An execution
of the algorithm is split into {\em stages}. In a stage, packets of total length (close to) $ck\ell_k$
are transmitted, where $c\in\NAT$ is a fixed large constant.
At the beginning of a stage, the set $C$ of candidates is determined as
$C=\{i\,|\, n_i\ell_i\geq ck\ell_k\}$.
Then, the {\em interesting} length $\ell_{i^\star}$ is set for parameter $i^\star=\min(C)$, and the algorithm
starts transmitting packets $l_{i^\star}$.
After each transmission, successful or not, the interesting length $i^\star$ is updated
to $i^\star\leftarrow \min(\{i^\star\}\cup\{i\,|\,\ell_i n_i\geq ck\ell_k\})$.
(Note that the set of candidates $\{i\,|\,\ell_i n_i\geq ck\ell_k\}$ may change over time,
as the adversary injects packets.)

Using the notion of the interesting length, we work in line with the original algorithm {\greedy},
with the following restrictions:
 \begin{itemize}
 \item
 no packet is transmitted as long as the interesting length is not determined (i.e., the set
 of candidates is empty);
 \item
 only a packet of length $l_{i^\star}$ can be transmitted.
 \end{itemize}

\begin{algorithm}[t]
	\caption{{\mgreedy}}
	\label{alg:mgreedy}
	\begin{algorithmic}[1]
    \State $C\gets \{i\,|\, n_i\ell_i\geq ck\ell_k\}$
    \Loop
        \While{$\{i\,|\,\ell_i n_i\geq ck\ell_k\}=\emptyset$} Stay {\em idle}
        \EndWhile
        \State $C\gets \{i\,|\,\ell_i n_i\geq ck\ell_k\}$
        \State $i^\star\gets \min(C)$
        \For{$a=1$ \textbf{to} $ck$}
                $\ell'\gets {\group}(j-1)$
                \State $\ell\gets \ell+\ell'$
        \EndFor
    \EndLoop
    \end{algorithmic}
\end{algorithm}
\begin{algorithm}[t]
	\caption{{\group}$(j)$}
	\label{alg:group2}
	\begin{algorithmic}[1]
    \State $\ell\gets 0$
    \While{$\ell\leq \ell_j-\ell_{i^\star}$}
        \If{$j>i^\star$}
                $\ell'\gets {\group}(j-1)$
                \State $\ell\gets \ell+\ell'$
        \Else
            \State Transmit $\ell_j$
            \State $C\gets C\cup \{i\,|\,\ell_i n_i\geq ck\ell_k\}$
            \State $i^\star\gets \min(C)$
            \State If a transmission of $\ell_j$ successful: $\ell\gets\ell_j$
        \EndIf
    \EndWhile
   \State \Return $\ell$
    \end{algorithmic}
\end{algorithm}
As the total length of packets staying in the queue whose lengths are not interesting
is at most $k\cdot ck\ell_k$, they do not have impact on the {\em asymptotic}
value of the relative throughput. Thus, assume that there are
no packets of lengths which are {\em not} interesting at each time $t$.
That is, there are no packets of lengths $\ell_i$ such that $i\not\in C$.
Then, the new algorithm {\mgreedy} works exactly as the original algorithm
{\greedy}. The pseudo-code of algorithm {\mgreedy} and the modified sub-routine
$\group(j)$, which now returns also some value $\ell$, are given as Algorithm~\ref{alg:mgreedy}
and Algorithm~\ref{alg:group2}, respectively.

\paragraph{Performance analysis of algorithm {\mgreedy}.}

We say that an execution of $\group(k)$ 
is {\em uniform} if the
algorithm transmits packets of a fixed length $\ell_i$ during that executions of $\group(k)$
as
well
as during the executions of $\group(k)$ directly preceding it.
A new key property of algorithm {\mgreedy} compared with {\greedy} is that
most of its executions of sub-routine $\group(k)$ are uniform.

\begin{proposition}\labell{prop:uniform}
At least $ck-2k$ calls of {\group} in a stage of {\mgreedy} are uniform.
\end{proposition}

\begin{proof}
Observe that the value of $i^\star$ can only decrease during a stage and, as long
as $i^\star$ remains unchanged, only packets of length $\ell_{i^\star}$ are transmitted.
Therefore, the claim follows from the fact that $i^\star$ may change at most
$k-1$ times during a stage.
\end{proof}
Now, we evaluate the relative throughput of {\mgreedy}.

\begin{lemma}\label{lem:klower}
The relative throughput of the  {\mgreedy} algorithm is at least
$$\min_{1\le j < i \le k}\left\{\frac{\lfloor\rho_{i,j}\rfloor}{\lfloor\rho_{i,j}\rfloor+\rho_{i,j}}\right\}\cdot\frac1{1+4/(c\eta)}
\ ,$$
where $\eta$ is a constant depending merely on packets' lengths.
\end{lemma}
\begin{proof}
In Lemmas~\ref{lem:k1}, \ref{lem:k2} and \ref{lem:k3}, we repeatedly use an argument
that, if {\greedy} does not use packets of length $\ell_i$ for $i>m$, then each
such packet transmitted by {\off} corresponds to a group of (shorter) packets transmitted
by {\greedy} of the total length $\ell_i$, and this association is injective, provided {\greedy} is
not idle at that time.
For a while,
assume that each execution of {\group} in {\mgreedy} is uniform.
Then, the above property of {\greedy}
can be preserved for the
{\mgreedy} algorithm with a relaxation that a packet $p$ of length $\ell_i$ transmitted by {\off}
corresponds  to a group of $\lfloor\ell_i/\ell_j\rfloor$ packets of length $\ell_j$ transmitted
by {\mgreedy} for $j<i$. Actually, those are the packets whose successful transmission is finished during
the transmission of $p$. The fact that there are at least $\lfloor\ell_i/\ell_j\rfloor$ such packets
follows from the assumption that {\mgreedy} is not idle at that time, its executions of {\group}
are uniform, and the time period needed for transmission of $\lfloor\ell_i/\ell_j\rfloor$ packets $\ell_j$
is not larger than $\ell_i$.

Let $\gamma=\min_{1\le j < i \le k}\left\{\frac{\lfloor\rho_{i,j}\rfloor}{\lfloor\rho_{i,j}\rfloor+\rho_{i,j}} \right\}.$
Let $\delta_{i,j}=\frac{\rho_{i,j}}{\lfloor \rho_{i,j}\rfloor}$ and let $\delta=\max_{i>j}\{\delta_{i,j}\}$.
One can check that $\gamma=1/(1+\delta)$.
This relaxation translates inequalities from Lemmas~\ref{lem:k1}, \ref{lem:k2} and \ref{lem:k3} to:
\begin{equation*}
\begin{array}{rcl}
\delta\cdot L_{\greedy}(\tau)&\geq& L_{\off}(\geq m, \tau)-\ell_k\\
(1+\delta)\cdot L_{\greedy}([t_1,t])&\geq & L_{\off}([t_1,t])\\
&&+q_{\off}(<m,t)-\ell_m-\ell_k\\
(1+\delta)\cdot L_{\greedy}(\tau)&\geq & L_{\off}(\tau)-f_m
\end{array}
\end{equation*}
If we apply the above inequalities instead of those from Lemmas~\ref{lem:k1}, \ref{lem:k2} and \ref{lem:k3}
in the proof of Theorem~\ref{th:kdivlower}, we obtain the claimed result, even without the $\frac1{1+4/(c\eta)}$ factor,
since
$1/(1+\delta)=\min_{1\le j < i \le k}\left\{\frac{\lfloor\rho_{i,j}\rfloor}{\lfloor\rho_{i,j}\rfloor+\rho_{i,j}}\right\}$.

However, the above reasoning does not deal with the situation that an execution of {\group} is not uniform.
In such a case, we cannot associate $\lfloor\ell_i/\ell_j\rfloor$
packets $\ell_j$ to $\ell_i$ transmitted by {\off} in the period $T$ such that their transmissions by {\mgreedy} were finished in $T$.
Fortunately, by Proposition~\ref{prop:uniform}, only the fraction $\frac{2k}{ck}=\frac2{c}$ of calls of $\group(k)$
are not uniform. As argued earlier in the proof of Lemma~\ref{lem:nospeed:aux}, even without uniformity assumption
we have the bound $L_{\greedy}/L_{\off}\geq \eta$, where
$\eta=\min_{i\in[2,k]}\left\{\frac{\rho_i-1}{2\rho_i-1} \right\}.$

Let us split packets transmitted by {\off}
by time $t$
into those whose transmission was inside periods when {\mgreedy}
works in a uniform manner, denoted by $L_{\off,1}(t)$, and the remaining ones, denoted by $L_{\off,2}(t)$. Given
the fact that the fraction at least $(1-\frac2{c})$ of transmitted packets by {\mgreedy} are sent in uniform
way, we have the following
$$\begin{array}{rcl}
\lim_{t\to\infty}(1-2/c)\cdot L_{\mgreedy}(t)/L_{\off,1}(t) &\geq& \gamma \\
\lim_{t\to\infty}(2/c)\cdot L_{\mgreedy}(t)/L_{\off,2}(t) &\geq& \eta
\end{array}$$
The second bound implies that
$(2/c)\cdot L_{\mgreedy}(t)/L_{\off,2}(t)\geq \eta/2$ for each large enough $t$, and therefore
$L_{\off,2}(t)\leq 4/(c\cdot\eta)L_{\mgreedy}(t)$. This implies also
that either $L_{\mgreedy}(t)>L_{\off}(t)$ or $L_{\off,2}(t)\leq 4/(c\cdot\eta) L_{\off}(t)$ for sufficiently large $t$.
As the first condition implies that the relative throughput is close to $1$, assume that $L_{\off,2}(t)\leq 4/(c\cdot\eta) L_{\off}(t)$.
As $L_{\off}(t)=L_{\off,1}(t)+L_{\off,2}(t)$, this implies that
$$\lim_{t\to\infty} L_{\mgreedy}(t)/L_{\off}(t) \geq \gamma\cdot \frac{1}{1+4/(c\eta)}\ .$$
\end{proof}
As we can choose arbitrarily large $c$, Lemma~\ref{lem:klower} implies that the relative throughput of {\mgreedy} might
be arbitrarily close to the upper bound from Theorem~\ref{th:nospeed:upper}.
In the following theorem, we argue that one can modify {\mgreedy} such that
it gradually increases the constant $c$ during its execution, which guarantees
the optimal relative throughput.

\begin{theorem}\labell{th:nospeedup:final}
The optimal relative throughput of an online algorithm is equal to
$$\min_{1\le j < i \le k}\left\{\frac{\lfloor\rho_{i,j}\rfloor}{\lfloor\rho_{i,j}\rfloor+\rho_{i,j}} \right\} \ .$$
\end{theorem}

\begin{proof}
Let $\gamma=\min_{1\le j < i \le k}\left\{\frac{\lfloor\rho_{i,j}\rfloor}{\lfloor\rho_{i,j}\rfloor+\rho_{i,j}} \right\}$.
By choosing sufficiently large constant $c$, algorithm {\mgreedy} can achieve the relative throughput which is arbitrarily
close to $\gamma$. More precisely, it is $\gamma_c=\gamma\cdot\frac1{1+4/(c\eta)}$.
From the proofs of Lemmas~\ref{lem:nospeed:aux} and \ref{lem:klower}, one can derive a polynomial
$p(1/\varepsilon,c,k,\ell_1,\ldots,\ell_k)$ such that
$$L_{\mgreedy}(t)\geq (1-\varepsilon)\gamma_c L_{\off}(t) \ ,$$
provided {\mgreedy} transmitted packets of length
at least $p(1/\varepsilon,c,k,\ell_1,\ldots,\ell_k)$.
Using this bound, one can design an adaptive version of {\mgreedy} which gradually increases the value
of its parameter $c$. The value of $c$ is increased to $2c$ when the total length of transmitted packets
is long enough to guarantee that the current relative throughput is close enough to $\gamma_c$ for
the current value of $c$ on one side, and deterioration of the relative throughput following from
the increase of $c$ in the initial part of the computation with larger $c$ is meaningless on the other side.
(Note that the increase of $c$ may cause a temporary deterioration of the relative throughput, since
larger $c$ requires more copies of $\ell_i$ in the queue to consider $\ell_i$ as an interesting packet's length.)
In this way, the relative throughput of the algorithm will get arbitrary close to $\gamma$
after sufficiently long time. (The actual bound on the current relative throughput at time $t$ depends rather on
the number of successfully transmitted packets than on time $t$.)
\comment{ 
For a constant $c'>0$, let $x_{c'}$ be such that if {\mgreedy} executed with the value $c\leftarrow c'$
transmits packets of total length $\geq x_{c'}$ up to $t$, then its relative throughput at $t$
is at least

It starts with $c\leftarrow c_0$ for some fixed $c_0$. Fix $\varepsilon_0>0$.

Then, the algorithm keeps transmitting packets
of total length $x$ which guarantee that
its relative throughput is at least $(1-\varepsilon_{0})\gamma_{c_0}$, and it remains at least
$(1-\varepsilon_{0})\gamma_{c_0}$ } 
\end{proof}

One can observe that, in order to get closer to the asymptotically optimal relative throughput,
our algorithms wait until there are many packets waiting for transmission in the queue (of total
length at least $ck^2\ell_k$ in the worst case). That is why we conjecture that the original
much simpler algorithm {\greedy}
might turn out to be more efficient in usual real life scenarios. Therefore, from
practical point of view, it is interesting to design an algorithm which achieves optimal (asymptotic)
relative throughput and minimizes the maximum of the relative throughput over all times $t$.

\section{An algorithm for a scenario with speedup}\label{sec:speedup}
\newcommand{\A}{Prudent}
Now we return to the packets whose lengths fulfil divisibility property, i.e.
$\ell_i/\ell_{i-1}\in\NAT$ for $1<i<k$, and address 
the problem of increasing throughput by enabling algorithm to work with  greater speed.
We design an algorithm {\A} which, working with speedup $s=2$, achieves relative throughput $1$.
This algorithm works in {\em phases}, where a phase is a time period between two consecutive errors.
Behaviour of the algorithm in a phase is described as Algorithm~\ref{alg:fast}.
During each phase it tries to send packets of maximal length which do not exceed the total length of packets sent so far. It can be treated as a greedy strategy restricted by a ''safety policy'' that does not allow to send  long packets unless the cost of their unsuccessful transmissions can be amortized by an advantage over an adversary gained during the earlier transmissions since the time of the last error.

\begin{algorithm}[t]
	\caption{{\A}}
	\label{alg:fast}
	\begin{algorithmic}[1]
    \Loop
        \While{ 
        $\{i\,|\, \ell_i n_i\geq \ell_k\}=\emptyset$}
            Stay {\em idle}
        \EndWhile
        \State  let $i$ be the smallest number such that: $n_i\ell_i\geq \ell_k$;
         \If {$i<k$}
            \State transmit $\ell_{i+1}/\ell_i$ packets $\ell_i$;
         				 \State $L_{sent} \gets \ell_{i+1}$
         				  \While {$L_{sent}<\ell_k$}
         				    \State $j \gets$ maximal number such that
                            \State\ \ \ \ \ \  $n_j\ell_j\geq \ell_k-L_{sent}$ and $\ell_j\leq L_{sent}$
         				    \State transmit $\ell_{j+1}/\ell_j$ packets $\ell_j$
         				    \State $L_{sent} \gets L_{sent}+\ell_{j+1}$
         		      \EndWhile
        \EndIf
        \Loop \State transmit longest unsent packet
        \EndLoop
    \EndLoop
    \end{algorithmic}
\end{algorithm}

\paragraph{Performance analysis of algorithm {\A}.}
Consider any offline algorithm $\off$.

\begin{lemma}\label{speeduplemma}
 The total length of packets sent by {\A} is less than the total length of packets sent by $\off$ by no more than $5/2k\ell_k$.
\end{lemma}

In the proof of the lemma we consider potential gain of $\off$ over {\A} restricted to $k-i$ longest types of packets for $i=k-1,\ldots,0$. 
The proof is inductive and it is splitted into next two propositions. The first one, being the base of the induction, shows that {\A} sends no less of the longest packets than $\off$. Then, in the second proposition, we make an inductive step. 
Here the notion of the $i$-th queue means the set of packets $\ell_i$ waiting for transmission in the queue of  algorithm {\A}.

\begin{proposition}
  $L_{\A}(k,t) \geq L_{\off}(k,t)$ for any time $t$.
\end{proposition}
\begin{proof}
%
%
 Let  $t$ be the earliest time in which  $L_{\A}(k,t) < L_{\off}(k,t)$.  There were no errors in the period $\tau=[t-\ell_k, t]$, so either {\A} has transmitted enough packets before $t-\ell_k/2$ to get willing to send packets $\ell_k$ or, for each $j<k$, the inequality  $n_j \ell_j<\ell_k$  was satisfied at time $t-\ell_k$. In both cases {\A} would send a packet $\ell_k$ in $\tau$ (if then there was any in its queue), which contradicts the choice of $t$.
\end{proof}

\begin{proposition}
  For any time $t$ and any $1\leq i <k$, if
	
$$L_{\A}(\geq j, t) \geq L_{\off} (\geq j, t) - 5/2(k-j)\ell_k, \mbox{ for all}\ \ j>i$$
  then
  $$L_{\A}(\geq i, t) \geq L_{\off} (\geq i, t) - 5/2(k-i)\ell_k. $$
\end{proposition}

\begin{proof}
  Let $t_b<t$ be the beginning of the phase with time $t$, let $\tau=[t_b,t)$. If $|\tau|=t-t_b<\ell_i$, then  $\off$ does not send any packet $\ell_i$, and thus the thesis is trivially fulfilled.
  Therefore assume that $|\tau|\geq \ell_i$. We distinguish two cases:
  \begin{itemize}
    \item[] {\em Case} 1. $L_{\A}(\tau)=0$,
    \item[] {\em Case} 2. $L_{\A}(\tau)>0$.
  \end{itemize}
  Note that in the first case $|\tau|<\ell_k/2$, so $L_{\off}(\geq i,\tau)<\ell_k/2$ and $L_{\off}(\geq i,t)<L_{\off}(\geq i,t_b)+ \ell_k/2$.
  On the other hand,
  $$\begin{array}{rcl}
  L_{\A}(\geq i,t)&=&L_{\A}(\geq i,t_b)+L_{\A}(\geq i,\tau)\\
  &=& L_{\A}(>i,t_b)+L_{\A}(i,t_b) \ .
  \end{array}$$
  The second equality holds since {\A} sends no packets in $\tau$.
	
	To estimate the right side of this equation note that {\A} has not tried to send a packet $\ell_i$ at $t_b$, so its $i$-th queue contained no more than $\ell_k/\ell_i$ packets. Therefore $L_{\A}(i,t_b)\geq L_{\off}(i,t_b)-\ell_k$. Combining this inequality with inductive bounds on $L_{\A}(>i,t_b)$ we get:
$$\begin{array}{rcl}
L_{\A}(\geq i,t)&\geq& L_{\off}(>i,t_b)-5/2(k-(i+1))\ell_k\\
&&+ L_{\off}(i,t_b)-\ell_k\\
&=&L_{\off}(\geq i, t_b)+3/2\ell_k - 5/2(k-i)\ell_k \\
&>& L_{\off}(\geq i, t)-5/2(k-i)\ell_k \ .   
\end{array}$$

  Let now  $L_{\A}(\tau)>0$. 
  Let packets of length $l_i$ or longer be called {\em long} and let packets shorter than $\ell_i$ be 
  called {\em short}.
  Since we are interested in estimating the volume of packets of length at least $\ell_i$ transmitted by {\A}, we check carefully what may happen after time $t_b+\ell_i/2$, i.e., at the moment when {\A} can start to transmit long packets.  We analyse separately three possible situations:
	\begin{enumerate}
	   \item[] {\em Subcase} 2.1. {\A} sends a small packet 
	   at some moment after  
	   $t_b+\ell_i/2$;
		 \item[] {\em Subcase} 2.2. after $t_b+\ell_i/2$ algorithm {\A} sends only long packets; 
		 \item[] {\em Subcase} 2.3. after $t_b+\ell_i/2$ algorithm {\A} does not complete transmission of any packet.
	\end{enumerate}

First, consider {\em Subcase} 2.1. Let  $t'$ be the time at which {\A} for the last time before $t$ started transmitting a packet $\ell_j$ for some $j<i$. Since {\A} sends packets in blocks of lengths $\ell_1,\ldots,\ell_k$ (recall that {\A} has speedup $2$ and $\ell_i/\ell_{i-1}\in\NAT$), we know that $t'\geq t_b+\ell_i/2$.    
Therefore we can conclude that  at time $t'$ the $i$-th queue contained less than $\ell_k/\ell_i$ packets and estimate $L_{\A}(i,t')$ by
	$L_{\A}(i,t')\geq L_{\off}(i,t')-\ell_k$. This together with the inductive hypothesis gives
	
	$$L_{\A}(\geq i, t')\geq L_{\off}(\geq i, t')- 5/2(k-i)\ell_k+3/2\ell_k \ .$$

During the period $[t',t]$, algorithm {\A} was successfully	transmitting  packets of length at least $\ell_i$ with the exception of the beginning of the period, when it was sending packet $\ell_j$, and possibly the end, since its transmission of the last packet could be stopped by an error. Thus, in this period $\off$  could send long packets of the total length at most $\ell_j/2+\ell_k$ larger than {\A}. So finally we have
$$\begin{array}{rcl}
L_{\A}(\geq i, t)\!\!\!\! &=&\!\!\!\! L_{\A}(\geq i, t')+L_{\A}(\geq i, [t't])\\
&\geq&\!\!\!\! L_{\off}(\geq i, t')- 5/2(k-i)\ell_k+ 3/2\ell_k\\
&&\!\!\!\! + L_{\off}(\geq i, [t'-t]) -3/2 \ell_k\\
&\geq&\!\!\!\! L_{\off}(\geq i, t)- 5/2(k-i)\ell_k
\ .
\end{array}$$

As for {\em Subcase} 2.2, note that in the period $\tau=[t_b,t]$, {\A} finishes successfully its last transmission after time $(t_b+t)/2$. Otherwise it would successfully send more long packets (remind that the length of each next packet chosen by {\A} does not exceed the total length of packets transmitted so far). Thus we have $L_{\A}(\geq i, \tau)>2(|\tau|/2-\ell_i/2)= |\tau|-\ell_i$. On the other hand $L_{\off}(\geq i, \tau)\leq|\tau|$. Now we use the divisibility property of the packet lengths and observe that both $L_{\A}(\geq i, \tau)$ and $L_{\off}(\geq i, \tau)$ are multiple of $\ell_i$. Therefore the lower bound on the value $L_{\A}(\geq i, \tau)$ is not less than the upper bound on $L_{\off}(\geq i, \tau)$, hence $L_{\A}(\geq i, \tau)\geq L_{\off}(\geq i, \tau)$.

The situation described in the third subcase can not happen as it is contradictory with our assumptions  
$|\tau|\geq \ell_i$ and  $L_{\A}(\tau)>0$, which  directly follow from the construction of algorithm {\A}.
\end{proof}

As a simple consequence of Lemma~\ref{speeduplemma} we get the following theorem.

\begin{theorem}
 The relative throughput of Algorithm {\A} working with speed-up $2$ is equal to $1$, provided $\ell_i/\ell_{i-1} \in \NAT$ for each $i\in[2,k]$.
\end{theorem}

\section{Conclusions}

We presented novel efficient and reliable algorithms for online scheduling of packets of
different lengths. The first protocol assures maximum possible throughput for any
arrival and jamming patterns, and additionally it guarantees to be no more than twice worse than
the throughput of any other scheduling algorithm run under the same patterns.
The second algorithm guarantees at least as high throughput
as the optimal one, when run with additional speed-up of $2$, i.e., with twice higher frequency.
It demonstrates that one can use available resources in a scalable way to improve
throughput for any arrival and jamming patterns, even the worst possible ones.

The considered framework is very general,
and therefore it leaves a number of open extensions for further study, both theoretical,
simulational and experimental.
For example, what is the relative throughput in case of ``average'' arrival patterns,
i.e., satisfying some stochastic constraints. In case of two packet lengths, it has been shown
in~\cite{Anta-SIROCCO-2013} that for some stochastic distributions the relative throughput
could be higher than $1/2$, and it would be interesting to give a complete characterization
of stochastic arrival case for arbitrary number of packet lengths.
Similarly, some restricted class of arrival and/or jamming patterns, e.g., motivated
by specific physical or mobility scenarios, could allow better
use of the channel. For such more specific settings, theoretical results could be also complemented
by simulations run for particular physical models.
Other extensions could involve packet deadlines, priorities and dependencies.

\bibliography{references,Bibliography}

\bibliographystyle{abbrv}

\end{document}